\documentclass{article}
\usepackage[a4paper, left=2cm, right=2cm]{geometry}
\usepackage{graphicx} % Required for inserting images
\usepackage{amsmath}
\usepackage{amsfonts}
\usepackage{amssymb}
\usepackage{xcolor}
\usepackage{algorithm}
\usepackage{algpseudocode}
\usepackage{stmaryrd}

\newtheorem{theorem}{Theorem}
\newtheorem{proposition}{Proposition}
\newtheorem{definition}{Definition}

\newtheorem{lemma}{Lemma}
\newtheorem{remark}{Remark}
\newtheorem{assumption}{Assumption}

\newcommand{\interval}[1]{\left[#1\right]}
\newcommand\intsimple[1]{\left\llbracket #1 \right\rrbracket}
\newcommand{\myvec}[2]{%
\left[
\begin{array}{c}
#1\\
#2
\end{array}
\right]
}
\newcommand{\ub}[1]{\bar{#1}}
\newcommand{\lb}[1]{\underline{#1}}
\newcommand{\II}{\mathcal{I}}
\newcommand{\MM}{\mathcal{M}}
\newcommand{\RR}{\mathbb{R}}
\newcommand{\TT}{\mathbb{T}}
\newcommand{\XX}{\mathcal{X}}
\newcommand{\UU}{\mathcal{U}}
\renewcommand{\SS}{\mathcal{S}}

\newcommand{\BB}{\mathcal{B}}
\newcommand{\Ac}{\hat{A}}
\newcommand{\Bc}{\hat{B}}
\newcommand{\AKc}{\hat{A}_K}
\newcommand{\AK}{\mathcal{A}_K}
\newcommand{\MD}{\II_S}
\newcommand{\MDK}{\II_{A_K}}
\newcommand{\Mdelta}{\II_{\Delta}}
\newcommand{\bbox}[1]{\mathbb{B}\left(#1\right)}
\newcommand{\qed}{\hfill$\square$}
\newcommand{\EE}{\mathcal{E}}
\title{Robust Model Predictive Control for Linear Systems\\ with Interval Matrix Model Uncertainty}
\author{Renato Quartullo\thanks{Uninettuno University, Rome, Italy. Email: renato.quartullo@uninettunouniversity.net}, Andrea Garulli\thanks{University of Siena, Siena, Italy. Email: andrea.garulli@unisi.it}, Mirko Leomanni\thanks{Mercatorum University, Rome, Italy. Email: mirko.leomanni@unimercatorum.it}}
\date{}

\begin{document}

\maketitle
\begin{abstract}
    This paper proposes a novel robust Model Predictive Control (MPC) scheme for linear discrete-time systems affected by model uncertainty described by interval matrices.
    The key feature of the proposed method is a bound on the uncertainty propagation along the prediction horizon which exploits a set-theoretic over-approximation of each term of the uncertain system impulse response.
    Such an approximation is based on matrix zonotopes and leverages the interval matrix structure of the uncertainty model.
    Its main advantage is that all the relevant bounds are computed offline, thus making the online computational load independent of the number of uncertain parameters.
    A variable-horizon MPC formulation is adopted to guarantee recursive feasibility and to ensure robust asymptotic stability of the closed-loop system.
    Numerical simulations demonstrate that the proposed approach is able to match the feasibility regions of the most effective state-of-the-art methods, while significantly reducing the computational burden, thereby enabling the treatment of nontrivial dimensional systems with multiple uncertain parameters.
\end{abstract}

\section{Introduction}

Model Predictive Control (MPC) is widely recognized as an effective technique to address multivariable control problems with state and input constraints \cite{rawlings2017model,borrelli2017predictive,CANNONbook}.
Robust MPC is an active research area that deals with the application of MPC to systems affected by uncertainty. 
The main challenges in this context is to provide theoretical guarantees on recursive feasibility,  robust stability and constraint satisfaction of the MPC scheme, in the presence of disturbances and unmodeled dynamics.

The case of bounded additive disturbances has been investigated in great detail. Minimizing the worst-case performance over all possible disturbance sequences requires the solution of min-max control problems \cite{bemporad2003min,buerger2014active}, whose computational complexity typically increases exponentially with the dimension of the state and input vectors. A more efficient, yet more conservative, approach is the so-called tube-MPC (see, e.g.,  \cite{lee1999constrained,chisci2002feasibility,mayne2005robust,rakovic2012homothetic}), in which robust constraint satisfaction is achieved by suitably tightening the nominal constraints to cope with the effect of disturbances on the predicted nominal trajectory. 

Parametric uncertainty in the dynamic model is usually more challenging to deal with, mainly because uncertainty propagation along the prediction horizon depends on the system states and inputs. Early works consider systems with uncertain impulse responses \cite{campo1987robust}. A commonly used uncertainty description assumes that the system matrices belong to a polytopic set, thus leading to min-max optimization problems aiming at minimizing the worst-case performance with respect to all admissible models. Notably, in \cite{bemporad2003min} it is shown that the MPC control law can be computed explicitly offline by using multiparametric programming, although the complexity of the solution rapidly grows even for systems of small dimension. Hence, the research has focused on finding a suitable compromise between conservatism and computational complexity. Methods based on ellipsoidal invariant sets and linear matrix inequalities have been proposed \cite{kothare1996robust,kouvaritakis2002efficient,munoz2015robust}. Conservatism reduction has been pursued by employing polytopic tubes and invariant sets, see, e.g.,  {\cite{lee2000robust,langson2004robust,fleming2014robust}. A detailed review of these approaches can be found in \cite[Chapter 5]{CANNONbook}. 
%Their main limitation is that the shape of the tube must be set a priori and does not depend on the way uncertainty propagates along the prediction horizon. 

More recently, researchers have turned their attention towards methods that employ linear time-varying state feedback controllers. These approaches are potentially advantageous for reducing conservativeness with respect to tube-based approaches, but make it more difficult to guarantee recursive feasibility of the MPC scheme. 
For example, in \cite{BORRELLI22} uncertainty propagation is expressed in terms of norm bounds on the perturbation of the nominal predicted trajectory, that can be computed offline but must hold for any admissible input sequence. In \cite{chen2024robust}, a system level synthesis approach is adopted to derive an uncertainty bound which depends on the controller parameters. Both methods employ a shrinking horizon MPC along with a backup control law to cope with the lack of recursive feasibility of the robust optimal control problem. Moreover, the number of constraints of such a problem depends on the number of vertices of the model uncertainty polytope, which typically grows exponentially with the number of uncertain parameters.
%and quadratically on the length of the prediction horizon. 
Hence, such techniques appear to be suitable to problems involving low-order systems with few uncertain parameters.

In this paper, a new robust MPC scheme is presented for linear discrete-time systems, affected by model uncertainties described by interval matrices. Such uncertain models are commonly used, for example, in set-membership identification \cite{belforte1993,cerone2011,casini2017}.
The main novelty of the proposed approach lies in the way the set-valued dynamics of the uncertainty propagation along the prediction horizon is expressed in terms of the nominal trajectory. Namely, this is done by computing a set-theoretic bound on each term of the uncertain system impulse response.
These bounds exploit the interval matrix model description to propagate uncertainty in the form of matrix zonotopes.
A key feature is that all the bounds can be computed offline, which makes the computational burden of the online optimal control problem independent of the number of uncertain parameters.
A variable-horizon MPC approach is adopted \cite{richards2006robust,quartullo2025robust}, by including the horizon length among the optimization variables of the optimal control problem. This allows one to achieve recursive feasibility and to guarantee that the state trajectories converge in finite time to a  suitable robust invariant set, designed to guarantee robust asymptotic stability of the control scheme.
Numerical simulations show that the proposed technique yields feasibility domains comparable to those provided by the method introduced in \cite{chen2024robust}, which was shown to outperform both tube-based approaches and the method in \cite{BORRELLI22}.
Remarkably, this is achieved with a much lower computational burden, which allows the proposed approach to be applied to systems of nontrivial dimension, with several uncertain parameters.

Summing up, the main contribution of the paper is a new robust MPC scheme for linear systems with model uncertainty described by interval matrices. Its main features are: i) a new set-theoretic over-approximation of model uncertainty propagation based on matrix zonotopes; ii) a significant reduction of the online computational burden, since the aforementioned approximation leads to constraints that do not depend on the number of uncertain parameters} iii) theoretical guarantees on recursive feasibility and robust stability, which leverage the interval matrix uncertainty representation and a variable-horizon optimal control problem.

The rest of the paper is organized as follows. Section \ref{sec:notation} introduces the notation adopted in the paper and several properties of interval matrices and zonotopes that are instrumental for the subsequent developments. The robust constrained control problem is formulated in Section \ref{sec:problem}. Section \ref{sec:intervalMPC} presents the proposed robust MPC algorithm and its main properties. The results of numerical simulations are discussed in Section \ref{sec:numerical}, while Section \ref{sec:conclusions} provides concluding remarks and future perspectives.

\section{Notation and preliminaries}
\label{sec:notation}

This section provides the notation, definitions, and fundamental properties of matrix sets used in the paper.

\subsection{Notation and set operations}

All the sets considered in the paper are sets of real matrices. 
The symbol $\oplus$ represents the Minkowski sum of sets. In particular, given two sets $\mathcal{A}$, $\mathcal{B}$, $\mathcal{A}\oplus\mathcal{B} = \{A+B|\,\,A\in\mathcal{A},B\in\mathcal{B}\}$. 
%and $\mathcal{A}\ominus\mathcal{B} = \{a|\,\, \{a\}\oplus\mathcal{B}\subseteq\mathcal{A}\}$. 
The set multiplication is defined as $\mathcal{A}\mathcal{B} = \{AB|\,\,A\in\mathcal{A},B\in\mathcal{B}\}$. Similarly,  $\mathcal{A}^j = \{A^j|\,\,A\in\mathcal{A}\}$.
Given sets $\mathcal{A}$, $\mathcal{B}$, $\mathcal{C}$, it holds
\begin{align}
    (\mathcal{A} \oplus \mathcal{B}) \mathcal{C} \subseteq 
    \mathcal{A}\mathcal{C}  \oplus  \mathcal{B}\mathcal{C}, \label{eq:setdp1} \\
    \mathcal{A} ( \mathcal{B} \oplus \mathcal{C}) \subseteq 
    \mathcal{A}\mathcal{B}  \oplus  \mathcal{A}\mathcal{C}.
    \label{eq:setdp2} 
\end{align}

Let $\II=\interval{\lb{I},\,\ub{I}} = \{M\in\RR^{p\times q}|\,\, \lb{I}\leq M \leq \ub{I}\}$ be an \emph{interval matrix}, where
$\lb{I},\ub{I}\in\RR^{p\times q}$ 
and the operator $\leq$ denotes element-wise inequality. The interval matrix $\II$ can be represented as $\II = C \oplus 
\intsimple{\Delta}$, where $C= (\ub{I} + \lb{I})/2$ is the \emph{center} matrix and $
\intsimple{\Delta} =  \interval{-\Delta,\Delta}$ is the 
symmetric part of $\II$, with $ \Delta= (\ub{I} - \lb{I})/2 \geq 0$.
A \emph{matrix zonotope} is defined as $
\MM = \Big\{ M_C \in \mathbb{R}^{p \times q} \;\Big|\; 
M = M_C + \sum_{i=1}^{e} G_i \beta_i,\; |\beta_i|\leq 1 \Big\},$
where $M_C$ denotes the center, and 
$ G_1,\, G_2,\, \ldots,\, G_e \in \mathbb{R}^{p \times q}$
are the \emph{generators}.
This matrix zonotope is shortly referred to as $\MM = \langle M_C;\,G_1,\, G_2,\, \ldots,\, G_e\rangle$. 

Given two interval matrices $\II_1 = C_1 \oplus  \intsimple{\Delta_1}$, $\II_2 = C_2 \oplus  \intsimple{\Delta_2}$, their sum is equal to
\begin{equation}
\label{eq:intsum}
\II_1 \oplus \II_2 = C_1+C_2 \oplus \intsimple{\Delta_1+\Delta_2}.
\end{equation}
Let 
$\MM = \langle M;\, G_1 \dots, G_g \rangle \subseteq \mathbb{R}^{p \times q}$, 
$\mathcal{N} = \langle N;\, H_1, \dots, H_e \rangle \subseteq \mathbb{R}^{p \times q}$
be two matrix zonotopes.  
Their Minkowski sum $\MM \oplus \mathcal{N}$ is the matrix zonotope
\begin{equation}
\label{eq:mink_MZ}    
\MM \oplus \mathcal{N} = \langle M + N;\, 
G_1, \dots, G_g,\, H_1, \dots, H_e \rangle.
\end{equation}

\begin{definition}\label{def:interval_product}
Let $\II_1 = {C_1} \oplus \intsimple{\Delta_1}$ and $\II_2 = {C_2} \oplus \intsimple{\Delta_2}$ be two interval matrices.  
The interval matrix product $\II_1 * \II_2$ is defined as
\begin{equation}
    \II_1 * \II_2 
    = {C_1} {C_2} 
    \oplus 
    \intsimple{
        |{C_1}| \Delta_2 
        + \Delta_1 |{C_2}| 
        + \Delta_1 \Delta_2
    }.
\end{equation}
\end{definition}
The interval matrix product in Definition~\ref{def:interval_product}
provides an over-approximation of the exact set product, that is \cite{rump2012fast}
\begin{equation}
\label{eq:intprod}
\II_1 \II_2 \subseteq \II_1 * \II_2.
\end{equation}
By using \eqref{eq:intprod},
the product of an interval matrix $\II = C \oplus \intsimple{\Delta}$ and a matrix $M$ can be over-approximated as
\begin{align}
& \II M = C M \oplus \intsimple{\Delta}M \;\subseteq\; C M \oplus \intsimple{\Delta |M|}, \label{eq:intxM}    \\    
& M \II  = M C \oplus M \intsimple{\Delta} \;\subseteq\; M C \oplus \intsimple{ |M| \Delta} .\label{eq:intxMleft}        
\end{align}
The following property holds~\cite{althoff2010reachability}.
\begin{proposition}\label{prop:box}
The minimum interval matrix containing a matrix zonotope 
$\MM = \langle M;\, G_1, \ldots, G_g \rangle$, denoted by $\bbox{\MM}$, is given by
$\bbox{\MM} = M \oplus \intsimple{\Delta}$, where $\Delta = \sum_{i=1}^g |G_i|.
$
\end{proposition}

\subsection{Set approximations}

In the following, some tools are provided which allow one to construct outer approximations of sets defined by the multiplication between interval matrices and matrix zonotopes.
Let us start by introducing a useful matrix decomposition.

\begin{definition}\label{def:ewdec}
    Given a matrix $ M \in \mathbb{R}^{l \times q} $, we define the \emph{entrywise decomposition} as the set of matrices $\mathcal{E}(M) = \left\{ M^{(k)} \right\}_{k=1}^{lq}$, where each $ M^{(k)} \in \mathbb{R}^{l \times q} $ is a matrix such that
$$
M^{(k)}_{ij} =
\begin{cases}
M_{ij}, & \text{if } k = (i - 1)l + j, \\
0, & \text{otherwise}.
\end{cases}
$$
\end{definition}
In words, $ \mathcal{E}(M) $ yields a sequence of matrices, each containing exactly one entry of $ M $ in its original position and zeros elsewhere.
Clearly, $\sum_{k=1}^{lq} M^{(k)} = M$. Notice that any interval matrix $\II = C \oplus 
\intsimple{\Delta}$ can be represented as the matrix zonotope $\langle c; \,\mathcal{E}(\Delta)\rangle$.  
Moreover, the following property holds.
\begin{proposition}\label{prop:AdecM}
Let $A\in\mathbb{R}^{n \times l}$ and $M \in \mathbb{R}^{l \times q}$, with $M \geq 0$. Then, 
\begin{equation}
\label{eq:AdecM}
\sum_{k=1}^{lq} |A M^{(k)} | =|A| M.
\end{equation}
\end{proposition}
\emph{Proof:} From Definition \ref{def:ewdec}, $|A M^{(k)} |= [0 \ldots 0~ |A_{| i}|M_{ij}~ 0 \ldots 0]$, where $A_{| i}$ denotes the $i$th column of A. Then,
$$
\sum_{k=1}^{lq} |A M^{(k)} | =\sum_{i=1}^{l}\sum_{j=1}^{q} |A_{| i}| M_{ij}
$$
which gives \eqref{eq:AdecM}.\qed\\
In the paper, the shorthand notation $A\mathcal{E}(M)$ is adopted to denote the set of matrices $AM^{(1)}, AM^{(2)}, \ldots, AM^{(lq)}$.

Let us introduce an operator that plays a key role throughout the paper.
\begin{definition}\label{def:operatorT}
Let $\II \in \mathbb{R}^{l \times p}$ be an interval matrix of the form
$ \II \;=\; C \oplus \intsimple{\Delta},$
and $\MM = \langle M;\, G_1, \ldots, G_g\rangle \subseteq \mathbb{R}^{p \times q}$ be a matrix zonotope. The operator $\TT_{\II}(\MM)$ is defined as
\begin{equation}\label{eq:Topdef}
    \TT_{\II}(\MM) = \langle CM;\, CG_1, \ldots, CG_g,\, F^{(1)}, \ldots, F^{(lq)} \rangle,
\end{equation}
where 
\begin{equation} \label{eq:Fgen}
F \;=\; {\Delta}\!\left( \,\lvert M \rvert + \sum_{j=1}^g \lvert G_j \rvert \,\right) \;\in\; \mathbb{R}^{l \times q}
\end{equation}
and the additional generators $F^{(i)}$, $i=1,\dots,lq$, in \eqref{eq:Topdef} are given by the entrywise decomposition of $F$, i.e.,
$\{F^{(i)}\}_{i=1}^{lq} = \mathcal{E}(F)$.
\end{definition}
The next result shows that the operator $\TT$ can be used to over-approximate the product between an interval matrix and a matrix zonotope.
\begin{proposition}\label{prop:overapprox}
Let $\II \in \mathbb{R}^{l \times p}$ be an interval matrix and $\MM\subseteq \mathbb{R}^{p \times q}$ be a matrix zonotope, then $\II\MM\subseteq\TT_{\II}(\MM).$
\end{proposition}
\emph{Proof. } Let $\II$ and $\MM$ be as in Definition \ref{def:operatorT}.
Then, by using \eqref{eq:setdp1}
\begin{equation}
\label{eq:IM}
%\begin{split}
\II \MM = (C \oplus \intsimple{\Delta}) \langle M;\, G_1, \ldots, G_g\rangle
\subseteq \langle C M;\, C G_1, \ldots, C G_g \rangle \oplus \intsimple{\Delta}\MM.
%\end{split}
\end{equation}
Since $\MM$ can be expressed as 
$$
\MM = M \oplus \langle 0;\, G_1, \ldots, G_g \rangle,
$$ 
by exploiting \eqref{eq:setdp2}, \eqref{eq:intxM} and Proposition \ref{prop:box}, one gets
\begin{equation*}
    \begin{split}
        \intsimple{\Delta}\MM & \subseteq \intsimple{\Delta}M \oplus\intsimple{\Delta}\langle0;\,G_1,\ldots,G_g\rangle\\
        %& \subseteq \intsimple{\Delta|M|} \oplus \left\{ M:\,\,M \in \sum_{i=1}^g \intsimple{\Delta}G_i\intsimple{1}\right\}\\
        & \subseteq \intsimple{\Delta|M|} \oplus \intsimple{\Delta} \intsimple{\sum_{i=1}^g |G_i|} \\
        & \overset{\eqref{eq:intprod}}{\subseteq} \intsimple{\Delta|M|} \oplus \intsimple{ \Delta \sum_{i=1}^g |G_i|}\\
        & \overset{\eqref{eq:intsum}}{=} \intsimple{\Delta\left(|M| + \sum_{i=1}^g |G_i|\right)}=\intsimple{F}= \langle 0; \EE(F) \rangle.
    \end{split}
\end{equation*}
Hence, the results follows from \eqref{eq:IM}, by using \eqref{eq:mink_MZ} and the definition of $\TT_{\II}(\MM)$ in \eqref{eq:Topdef}.
\qed

\begin{definition}\label{def:jfold}
    For $j \in \mathbb{N}$, we define the $j$-fold application of the operator $\TT_{\II}$ as
$$
    \TT_{\II}^j(\MM) \;=\; 
    \underbrace{\TT_{\II} \circ \TT_{\II} \circ \cdots \circ \TT_{\II}}_{j \text{ times}} (\MM) 
    %= \TT_{\II}\circ\TT_{\II}^{j-1}(\MM).
$$
\end{definition}
By iteratively applying Proposition~\ref{prop:overapprox}, one has 
\begin{equation}
\label{eq:IjMapprox}
\II^j\MM\subseteq\TT_{\II}^j(\MM).
\end{equation}

\begin{proposition}\label{prop:lesscons}
    Let $\II\in \mathbb{R}^{l \times p}$ and $\MM\in \mathbb{R}^{p \times q}$ denote an interval matrix and a matrix zonotope, respectively. Then, the following inclusion holds
    \begin{equation}
        \bbox{\TT_{\II}(\MM)} \subseteq \II * \bbox{\MM}.
    \end{equation}
\end{proposition}
\emph{Proof:} From Proposition~\ref{prop:box}, Definition~\ref{def:operatorT}, and the notation introduced therein, it follows that
\begin{equation*}
    \begin{split}
        \bbox{\TT_{\II}(\MM)} & = \bbox{\langle CM;\, CG_1, \ldots, CG_g,\, F^{(1)}, \ldots, F^{(lq)} \rangle}\\
        & = CM\oplus\intsimple{\sum_{i=1}^{g}|CG_i|+\sum_{i=1}^{lq} |F^{(i)} |}\\
        &=  CM\oplus\intsimple{\sum_{i=1}^{g}|CG_i|+\Delta|M| +\Delta\sum_{i=1}^{lq} |G_i |}\\
        & \subseteq CM \oplus\intsimple{|C|\sum_{i=1}^{g}|G_i|+\Delta|M| +\Delta\sum_{i=1}^{lq} |G_i |}\\
        & = \II * \bbox{\MM}.
    \end{split}
\end{equation*}\qed

\section{Problem Formulation}
\label{sec:problem}
Consider the linear time-invariant system
	\begin{equation}\label{eq:sys}
		x(k+1) = Ax(k) + Bu(k),
	\end{equation}
where $x(k) \in\RR^n$ is the system state at time $k$, $u(k) \in\RR^m$ is the control input. The system matrices $A$ and $B$ are unknown and belong to an interval matrix set, i.e., 
\begin{equation}
\label{eq:sys2}
\interval{A\quad B}\in\MD =\interval{\Ac\quad\Bc}\oplus \intsimple{\Delta_A\quad \Delta_B},
\end{equation}
with known $\Ac,\,\Bc,\,\Delta_A,\,\Delta_B$. Let us denote $\Mdelta=\intsimple{\Delta_S}$ with $\Delta_S=\left[\Delta_A\quad\Delta_B\right]$. Then, system~\eqref{eq:sys} can be rewritten as $$x(k+1) = \Ac x(k) + \Bc u(k) + A_{\Delta} x(k) + B_{\Delta} u(k),$$
where $A_{\Delta} = A-\Ac$ and $B_{\Delta} = B-\Bc$ and $\left[A_{\Delta}\quad B_{\Delta}\right]\in \Mdelta$.
The control objective is to drive the state to the origin despite the model uncertainty. Moreover, the state and input are subject to the constraints
\begin{equation}\label{eq:constraint}
    x(k) \in \XX,\,u(k) \in \UU,
\end{equation}
with $\XX$, $\UU$ being convex polytopic sets.
Consider now the nominal system, given by 
\begin{equation}\label{eq:sys_nominal}
    z(k+1) = \Ac z(k) + \Bc v(k).
\end{equation}
Let the feedback policy for system~\eqref{eq:sys} be designed as
\begin{equation}\label{eq:u_RMPC}
    u(k) = K(x(k)-z(k)) + v(k),
 \end{equation} 
where $v(k)$ is a nominal control input and $K$ is a matrix gain satisfying the following assumption.
\begin{assumption}\label{assum:K}
    The matrix gain $K$ is such that $A_K = A+B K$ is Schur for all $\left[A\quad B\right]\in\MD$.
\end{assumption}
A feedback gain $K$ satisfying Assumption~\ref{assum:K} can be computed using synthesis methods based on linear matrix inequalities (LMI) (see, e.g., \cite{boyd1994lmi,kothare1996robust,CANNONbook}).
Clearly, also the matrix $A_K$ is uncertain and one has 
\begin{equation}
\label{eq:cAK}
    A_K\in\AK = \AKc \oplus \Mdelta\myvec{I}{K},
\end{equation}
where $\AKc = \Ac+\Bc K$.

The error between the true and the nominal state is defined as $e(k) = x(k)-z(k)$ and its dynamics evolve as
%\begin{equation}\label{eq:error_dyn}
%    e(k+1) = \AKc e(k) + A_{\Delta} x(k) + B_{\Delta} u(k)
%\end{equation}
%or equivalently
\begin{equation}\label{eq:error_dyn2}
    e(k+1) = A_K e(k) + A_{\Delta} z(k) + B_{\Delta} v(k).
\end{equation}
Therefore, due to the uncertainty affecting $A$ and $B$, the error $e(k)$ evolves within a set sequence $\SS(k)$, i.e. $e(k)\in\SS(k)$, where $\SS(k)$ satisfy the set-valued dynamics
\begin{equation}\label{eq:Sk_recursion}
    \SS(k+1) = \AK\SS(k) \oplus \Mdelta\left[\begin{array}{c}
         z(k)\\
         v(k)
    \end{array}\right].
\end{equation}

In this work, the aim is to compute the nominal input $v(k)$  by solving a variable-horizon optimal control problem defined as 
\begin{equation}\label{eq:mpc_tube_exact} 
    \begin{aligned}
        \underset{N_k,\textbf{v}(k),\textbf{z}(k)}{\text{min }} &   J_k(N_k,\textbf{v}(k),\textbf{z}(k))\\
        \text{s.t.} \quad & z_k(0)= x(k)\\
        & {z_k}(j+1)=\Ac {z_k}(j)+\Bc {v_k}(j)  \\
        &  {z_k}(j)\oplus\SS_k(j) \subseteq \XX, \quad j=0,\ldots,N_k \\
        &  v_k(j)\oplus K\SS_k(j)\subseteq \UU, \quad j=0,\ldots, N_k-1 \\
        &  z_k(N_k) \oplus \SS_k(N_k) \subseteq\XX_{f} \\
        & N_k \in \mathbb{N}^+
    \end{aligned}
\end{equation}
where the sequences $\boldsymbol{v}_k= \left(v_k(0),\ldots,\,v_k(N_k-1)\right)$ and $\boldsymbol{z}_k = \left(z_k(0),\ldots,\,z_k(N_k)\right)$ are the nominal inputs and states, respectively, along the prediction horizon of length $N_k$. 
Then, as it is usual in MPC, one sets $v(k)=v_k(0)$ and the process is repeated at every time step.
In \eqref{eq:mpc_tube_exact}, the cost $J_k$ is defined as 
\begin{equation}\label{eq:cost}
    J_k = \gamma N_k + \sum_{j=0}^{N_k-1} \ell(v_k(j)-Kz_k(j))
\end{equation}
where $\gamma> 0$ and $\ell(\cdot)$ is a positive function, i.e., $\ell(x)\geq 0,\,\forall x$. In accordance with~\eqref{eq:error_dyn2}, the prediction error sets in~\eqref{eq:mpc_tube_exact} are defined as 
\begin{equation}\label{eq:setdyn}
\SS_k(j+1) = \AK\SS_k(j)\oplus\Mdelta\left[\begin{array}{c}
     z_k(j)\\
     v_k(j) 
\end{array}\right]
\end{equation}
with the initial set being $\SS_k(0) = \{0\}$, due to the initial constraint $z_k(0) = x(k)$.
The sets $\SS_k(j)$ can be rewritten as
\begin{equation}\label{eq:Sk_response}
    \SS_k(j) =  \sum_{i=0}^{j-1}\AK^{j-i-1}\Mdelta\left[\begin{array}{c}
         z_k(i)\\
         v_k(i)
    \end{array}\right].
\end{equation}
With the aim to ensure robust stability of the control scheme, the terminal set $\mathcal{X}_f$ is designed as follows.
\begin{assumption}\label{assum:RPI}
    The terminal set $\XX_f$ is robust positive invariant (RPI) for the uncertain dynamics $ x(k+1) = A_K x(k) $ for all $A_K\in\AK$, i.e., $(A+BK)x\in\XX_f$, $\forall x\in \XX_f,\,\forall \left[A\quad B\right]\in\MD$. This is equivalent to require that 
    \begin{equation}\label{eq:RPI} 
        \AK\XX_f \subseteq\XX_f.
    \end{equation}
    Moreover, the set $\XX_f$ is such that
    \begin{subequations}
        \begin{align}
            \XX_f&\subseteq\XX, \label{subeq:stateconstr} \\
            K\XX_f&\subseteq\UU. \label{subeq:inputconstr}
        \end{align}
    \end{subequations}
\end{assumption}

\begin{remark}
    Methods for computing or approximating robust positive invariant sets are available in the literature. These include iterative algorithms for polytopic or interval-bounded systems \cite{kolmanovsky1998theory},~\cite[Ch.~10]{borrelli2017predictive}. LMI-based approaches for constructing RPI sets with ellipsoidal, polytopic or parallelotopic shapes are also available as discussed in \cite[Ch.~5]{CANNONbook}. Such methods can be used to determine a set $\XX_f$ satisfying Assumption~\ref{assum:RPI}.
\end{remark}

The rationale behind problem~\eqref{eq:mpc_tube_exact} is that uncertainty propagation through the sets $\mathcal{S}_k$ guarantees robust constraint satisfaction along the prediction horizon. Moreover, the variable-horizon approach adds further flexibility to the optimal control problem, which can be exploited to guarantee recursive feasibility. Unfortunately, solving problem~\eqref{eq:mpc_tube_exact} is practically intractable because the exact computation of the error sets $\SS_k(j)$ is prohibitive. This is due both to the need for multiple exact set multiplications, and to their dependence on the optimization variables $\mathbf{v}_k$ and $\mathbf{z}_k$, which prevents their offline pre-computation.
%Moreover, the set inclusions involved in the state and input constraints are highly impractical for standard solvers, since in general they require the computation of the vertices of polytopic sets. 
To address these issues, in the next section, we propose a reformulation of problem~\eqref{eq:mpc_tube_exact} that renders it efficiently solvable online.

\section{Robust MPC with Interval Matrices}
\label{sec:intervalMPC}
The main idea of the reformulated control scheme is to replace the exact error sets $\SS_k(j)$ appearing in the state and input constraints of problem~\eqref{eq:mpc_tube_exact} with a suitable linear combination of interval matrices that can be computed offline. To this aim, the control problem~\eqref{eq:mpc_tube_exact} is modified as follows.
\begin{subequations}\label{eq:mpc_tube_box}
    \begin{align}
        \min_{N_k,\mathbf v(k),\mathbf z(k)}\ & J_k(N_k,\mathbf v(k),\mathbf z(k)) \notag\\
        \text{s.t.}\quad & z_k(0)= x(k) \label{subeq:mpc_init}\\
        & z_k(j+1)=\Ac z_k(j)+\Bc v_k(j), \quad j=0,\ldots,N_k-1 \label{subeq:mpc_dynamic}\\
        & z_k(j)\oplus\BB_k(j) \subseteq \XX, \quad j=0,\ldots,N_k \label{subeq:mpc_state}\\
        & v_k(j)\oplus K\BB_k(j)\subseteq \UU, \quad j=0,\ldots, N_k-1 \label{subeq:mpc_input}\\
        & z_k(N_k) \oplus \BB_k(N_k) \subseteq \XX_f \label{subeq:mpc_terminal}\\
        & N_k \in \mathbb{N}^+ \notag
    \end{align}
\end{subequations}
%where the cost is the same as in~\eqref{eq:cost}. 
The sets $\BB_k(j)$ in~\eqref{subeq:mpc_state}-\eqref{subeq:mpc_terminal} are defined as 
\begin{equation}\label{eq:Bbox}
    \BB_k(j) = \sum_{i=0}^{j-1} \II(j-i-1)\myvec{z_k(i)}{v_k(i)},
\end{equation}
where
\begin{equation}\label{eq:Ibox}
    \II(j) =  \bbox{\TT^j_{\MDK}(\II_{\Delta})},
\end{equation}
and 
$\MDK = \AKc \oplus \intsimple{\Delta_K}$,
with $\Delta_K=\Delta_A + \Delta_B |K|$.
Denoting by $N^*_k$, $\mathbf{v^*(k)}$, $\mathbf{z^*(k)}$, the optimal solution of problem~\eqref{eq:mpc_tube_box}, the nominal control input in \eqref{eq:u_RMPC} is then chosen as
\begin{equation}\label{eq:optvsel}
    v(k)=v^*_k(0).
\end{equation}
The following property holds.
\begin{proposition}
\label{prop:Bk}
The sets $\BB_k(j)$ in \eqref{eq:Bbox} are such that    $\SS_k(j) \subseteq \BB_k(j),\,\forall k,j$.
\end{proposition}
\emph{Proof.} 
From \eqref{eq:cAK} and \eqref{eq:intxM}, one has that $\AK  \subseteq \MDK$. Hence, as a consequence of \eqref{eq:IjMapprox} one gets $\AK^{j}\Mdelta \subseteq \II(j)$ in \eqref{eq:Ibox}. Therefore, the result follows from definitions \eqref{eq:Sk_response} and \eqref{eq:Bbox}.\qed

Proposition \ref{prop:Bk} implies that constraints \eqref{subeq:mpc_state}-\eqref{subeq:mpc_input} are a conservative version of the corresponding constraints in problem \eqref{eq:mpc_tube_exact}. Notice that the sets $\BB_k(j)$ explicitly depend on the nominal state and input sequences $\textbf{z}_k$ and $\textbf{v}_k$, which are optimization variables of problem \eqref{eq:mpc_tube_box}. On the other hand, it is worth remarking that the interval matrices $\II(j)$ can be computed offline for all $j$, since they do not depend on $\textbf{v}_k$ and $\textbf{z}_k$. This is a key feature that allows to remarkably reduce the online computational burden.

\begin{remark} 
The use of the operator $\TT$ in \eqref{eq:Ibox} is crucial in order to reduce the conservatism in the approximation of the sets $\SS_k(j)$. In fact, the proposed approximation combines zonotopic and interval representations by exploiting Proposition~\ref{prop:overapprox} and \eqref{eq:IjMapprox}. Specifically, the error dynamics are propagated using the operator $\TT$ 
instead of the true set-valued dynamics~\eqref{eq:setdyn}. The resulting over-approximating matrix zonotope is then bounded by a tight interval matrix. In this way, according to Proposition~\ref{prop:lesscons}, the sets $\BB_k(j)$ provide an approximation of the true error sets $\SS_k(j)$ which is less conservative than that one would obtain by directly propagating interval matrices along the prediction horizon, using the interval matrix product operator in Definition~\ref{def:interval_product}. 
\end{remark}

Let us now establish recursive feasibility and convergence properties of the MPC scheme based on problem~\eqref{eq:mpc_tube_box}.

\begin{assumption}\label{assum:P0}
    Problem~\eqref{eq:mpc_tube_box} is feasible at time $k=0$ with associated optimal cost $J_0^*$.
\end{assumption}

\begin{theorem}\label{thm:RF}
Let Assumptions~\ref{assum:K}-\ref{assum:P0} be satisfied, then problem~\eqref{eq:mpc_tube_box} admits a feasible solution for all $k>0$ and for any $\left[A\quad B\right]\in\MD$.
\end{theorem}
\emph{Proof.} See Appendix~\ref{app:RF}.\qed

\begin{theorem}\label{thm:convergence}
    Let Assumptions~\ref{assum:K}-\ref{assum:P0} be satisfied. Then the following statements hold:
    \begin{itemize}
        \item[(i)] The closed-loop trajectories $x(k)$ of the uncertain system~\eqref{eq:sys}-\eqref{eq:sys2}, under the control law \eqref{eq:u_RMPC},\eqref{eq:mpc_tube_box}-\eqref{eq:optvsel}, reach the set $\XX_f$ in at most $\lfloor J_0^*/\gamma\rfloor$+1 steps.
        \item[(ii)] The closed-loop system~\eqref{eq:sys}-\eqref{eq:sys2},\eqref{eq:u_RMPC},\eqref{eq:mpc_tube_box}-\eqref{eq:optvsel} is robustly asymptotically stable. 
    \end{itemize}
\end{theorem}
\emph{Proof.} See Appendix~\ref{app:convergence}.\qed

Set constraints in \eqref{subeq:mpc_state}-\eqref{subeq:mpc_terminal} can be efficiently handled as follows.
The state constraints are given by the polytopic set $\XX = \{x\in\RR^n|\,\,H x\leq b\}$. Since all interval matrices $\II(j)$ in~\eqref{eq:Ibox}  are symmetric (i.e., the center is zero), they can be written as $\II(j) = \intsimple{\Delta_I(j)}$. Then, the $j$-th state constraint~\eqref{subeq:mpc_state}, i.e.,
$$z_k(j) \oplus \sum_{i=0}^{j-1}\II(j-1-i)\myvec{z_k(i)}{v_k(i)}\subseteq \XX,$$ 
can be equivalently expressed as
\begin{equation}\label{eq:constr_efficient}
    H z_k(j) + \left|H \right|\sum_{i=0}^{j-1} \Delta_I(j-i-1)\left|\myvec{z_k(i)}{v_k(i)}\right|\leq b.
\end{equation}
The input and terminal constraints~\eqref{subeq:mpc_input}-\eqref{subeq:mpc_terminal} can be treated similarly (assuming that also $\XX_f$ is a polytopic set).
The constraint~\eqref{eq:constr_efficient} can be easily cast as a linear constraint (see, e.g.,~\cite{boyd2004convex}) and thus it can be efficiently handled online by standard solvers. This is enabled by the use of interval matrices in the optimization problem~\eqref{eq:mpc_tube_box} to bound the system error. 

In view of \eqref{eq:constr_efficient}, problem~\eqref{eq:mpc_tube_box} turns out to be a mixed-integer program with linear constraints. 
For such type of problems, efficient heuristics have been proposed to speedup computations~\cite{leomanni2022variable,persson2024optimization}.

\begin{remark}\label{remark:vertices}
The proposed set-theoretic uncertainty bounding procedure does not require to impose
constraints on the vertices of the 
uncertainty set $\MD$. Actually, the constraints of problem \eqref{eq:mpc_tube_box} do not even
depend on the number of 
uncertain parameters. This leads to a remarkable reduction of the online
computational burden with respect to both 
tube-based approaches and to techniques which explicitly bound uncertainty
propagation along the prediction horizon, as 
\cite{BORRELLI22,chen2024robust}.
Moreover, with respect to the latter techniques, the proposed robust optimal control problem is guaranteed to be
recursively feasible.
\end{remark}

\section{Numerical results}
\label{sec:numerical}

In this section, three numerical examples are presented to assess the performance of the proposed Interval Matrix MPC approach, hereafter denoted as IM-MPC.
The first is a simple double-integrator system, used to illustrate the resulting feasible domain. 
The second example, borrowed from \cite{chen2024robust}, is used to evaluate the conservatism of the proposed control law with respect to the uncertain parameter bounds.
The third example involves a population model as a higher-dimensional benchmark, to demonstrate the computational efficiency of the method. 
The IM-MPC control law is compared to the following robust MPC schemes: 
the {Polytopic Tube-MPC (PT-MPC)} presented in~\cite[Sec.~5.5]{CANNONbook}; 
the Offline Tightening MPC (OT-MPC) approach introduced in~\cite{BORRELLI22}\footnote{For the comparison with~\cite{BORRELLI22}, the publicly available code at: https://github.com/monimoyb/RMPCPy has been used.}; 
the System Level Synthesis MPC (SLS-MPC) proposed in~\cite{chen2024robust}\footnote{For the comparison with~\cite{chen2024robust}, the publicly available code at: https://github.com/ShaoruChen/Polytopic-SLSMPC has been used.}. For IM-MPC, the function $\ell(\cdot)$ in the cost $J_k$ is quadratic, i.e., $J_k = \gamma N_k+\sum_{j=0}^{N_k-1} \|v_k(j)-Kz_k(j)\|_2^2$ with $\gamma = 1$. The variable-horizon problem~\eqref{eq:mpc_tube_box} is solved by enumeration for each horizon length $N_k \in\{ 1,\ldots,N_{\max}\}$.
The simulations are carried out using MATLAB on a standard laptop equipped with an Intel i5 processor (4.6~GHz) and 16~GB of RAM. The code is publicly available at {https://github.com/RenatoQuartullo/Interval-Matrices-MPC}.

\subsection{Feasible domain evaluation}\label{subsec:toy}
Consider a discrete-time uncertain double-integrator system,
in which 
%$x(k) =  \left[x_1(k),\,x_2(k)\right]^T$,
the nominal system matrices are set as
	\begin{equation*}\label{eq:AB_DI}
		\Ac=
		\left[
		\begin{array}{c c}
			1&\quad 1\\
			0&\quad 1 \end{array}\right],
		\quad
        \Bc=\left[
		\begin{array}{lll}
			0\\
			1\end{array}\right]
	\end{equation*}
    while the interval matrix uncertainty $\MD$ is given by
    \begin{equation*}
        \Delta_A = \left[
		\begin{array}{c c}
			0.1&\quad 0.05\\
			0.01&\quad 0.03 \end{array}\right], \quad
        \Delta_B=\left[
		\begin{array}{lll}
			0.05\\
			0.02\end{array}\right].
    \end{equation*}
The state constraint $\mathcal{X}$ is such that $$\myvec{-12}{-4}\leq x(k) \leq\myvec{12}{4},$$ while the input command must satisfy $|u(k)|\leq 2$. The matrix gain $K$ and the terminal set $\XX_f$ satisfying Assumptions~\ref{assum:K}-\ref{assum:RPI}
are calculated following the procedures described in~\cite[Sec.~5.2]{CANNONbook} and~\cite[Sec.~5.4]{CANNONbook}, respectively. The resulting gain matrix is $K=\left[-0.47\,\,-\!1.48\right]$ and the terminal set is $\XX_f = \{x:\,\,||Vx||_{\infty}\leq\alpha\}$, with $$V = \myvec{2.08\quad 2.07}{1.25\quad 2.65},\quad \alpha =\myvec{4.71}{1.48}.$$  
%The cost is chosen as in~\eqref{eq:cost} with $\ell(\cdot)$ being the 1-norm and $\gamma = 1$. 
%The variable-horizon problem~\eqref{eq:mpc_tube_box} is solved by enumeration for each $N_k = 1,\ldots, N_{\max}$, with $N_{\max}=25$.

Figure~\ref{fig:ROA} depicts the feasible domain obtained with the different control laws. 
The feasible domains are computed as the convex hull of the initial conditions for which 
problem~\eqref{eq:mpc_tube_box} is feasible at $k=0$. A total of 300 initial conditions are 
generated by uniformly gridding the admissible state set~$\XX$. 
To ensure a fair comparison, the feasible domains of PT-MPC and SLS-MPC are the computed as the set of all initial states for which the  optimization problem is feasible for at least one horizon length $N \in \{1,\ldots,N_{\max}\}$, with $N_{\max} = 25$.  
% This prevents any bias arising from the fact that the proposed scheme employs a variable prediction horizon, whereas the method in~\cite{CANNONbook} relies on a fixed horizon. 
%In this way, the comparison is not affected by a potentially bad choice of the horizon length. 
For OT-MPC, the maximum horizon length has been set to $N=3$, as longer horizon lengths required excessively high computation times. %\mycomm{Ma dire (come era prima) che gli autori stessi dicono questo?} 
% As reported by the authors, the control scheme cannot be executed reliably for prediction horizons larger than three. 
% In this evaluation, we used their public code\footnote{Available at: https://github.com/monimoyb/RMPCPy} with the setting described in this section.
It can be seen that the proposed IM-MPC yields a significantly larger feasible domain 
than PT-MPC and OT-MPC. 
Notice that the larger feasible domain obtained with the proposed control law is not only a consequence of the variable prediction horizon, since also the feasible domain of PT-MPC is evaluated for all horizon lengths up to $N_{\max}$. 
This confirms that the improvement is due to a less conservative treatment of uncertainty propagation along the prediction horizon. 
On the other hand, the feasible domain of  SLS-MPC is exactly equal to that achieved by the proposed method. This turns out to be equal also to the maximal robust control invariant set, meaning that no larger feasible domain can be achieved~\cite{borrelli2017predictive}.

%\mycomm{This comparable performance, however, this comes with a computational cost significantly higher than that of IM-MPC. In fact, the average runtime required to solve the SLS-MPC optimization problems, with the shortest feasible horizon for each initial condition, is 3.92~s. By contrast, the IM-MPC technique requires only 0.19~s, despite solving a sequence of up to $N_{\max}$ optimization problems to handle the variable-horizon formulation. The average runtime for the fixed-horizon PT-MPC is 0.23~s, thus comparable to that obtained by the IM-MPC.} \mycommRQ{Questo in rosso non si può più dire se si usa optimizer...i tempi vengono piuttosto simili.}

\begin{figure}[h]
    \centering
    \includegraphics[width=0.6\linewidth]{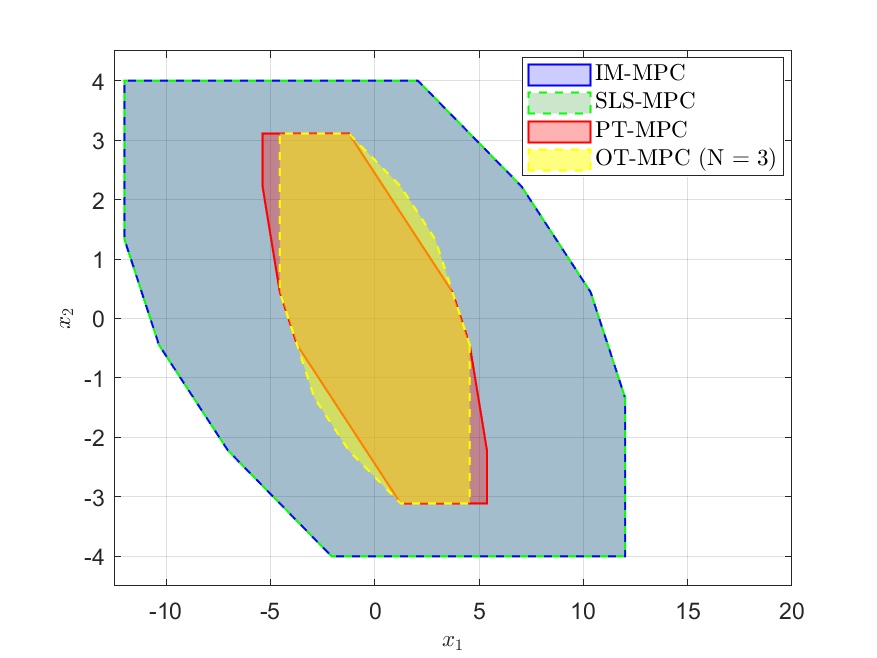}
    \caption{Feasible domain comparison for the double integrator example.}
    \label{fig:ROA}
\end{figure}

\subsection{Conservatism evaluation}\label{subsec:parmar}
In order to assess the degree of conservatism of the proposed uncertainty propagation procedure, let us consider the example presented in \cite[Section 6.2]{chen2024robust}. System \eqref{eq:sys} is defined by the matrices
\begin{equation*}\label{eq:AB_DI}
		A=
		\left[
		\begin{array}{c c}
			a &\quad 0.15\\
			0.1 &\quad 1 \end{array}\right],
		\quad
        B=\left[
		\begin{array}{lll}
			1\\
			b\end{array}\right]
\end{equation*}
with $1-\delta \leq a \leq 1+\delta$ and $1 \leq b \leq 1.2$. The state and input constraints are $|x_i(k)| \leq 8$, $i=1,2$, and  $|u(k)|\leq 4$. For increasing values of $\delta>0$, the maximal robust control invariant set for system \eqref{eq:sys}-\eqref{eq:constraint} is computed and then sampled to generate a grid of initial states. Then, feasibility of the robust optimal control problems of both SLS-MPC and IM-MPC is checked for each initial condition. 
For IM-MPC, the matrix $K$ is computed as in the example in Section \ref{subsec:toy} and we set $N_{\max} = 10$.
%In order to ensure a fair comparison, we consider SLS-MPC feasible if feasibility is obtained for at least one horizon length, among those considered in the variable-horizon problem of IM-MPC (namely, $1 \leq N \leq 20$). \mycomm{vero?}
The coverage rate of the feasible domain, measured as the ratio between the number of feasible initial states and the total number of sampled initial conditions, is reported in Figure \ref{fig:coverage} for all the techniques. Two prediction horizon lengths are considered for the SLS-MPC and PT-MPC schemes, namely $N = 3$ and $N = 10$ (as in~\cite{chen2024robust}), while for the OT-MPC only the case $N = 3$ is viable.  It can be observed that the conservatism degree of SLS-MPC and IM-MPC is essentially the same, and significantly lower than that of the other tested approaches.
%, with a slight advantage of the latter for higher values of $\delta$. 
This is quite remarkable, as in \cite{chen2024robust} it was shown that in this example SLS-MPC outperforms several other tube-based MPC techniques (the interested reader is referred to \cite{chen2024robust} for more details on such comparisons).
\begin{figure}
    \centering
    \includegraphics[width=0.5\linewidth]{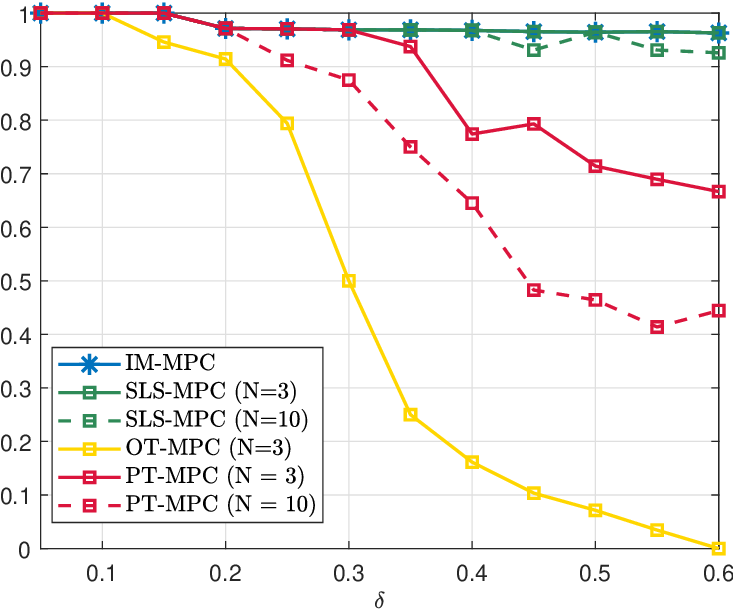}
    \caption{Coverage of the feasible domain for different uncertainty bounds $\delta$.}
    \label{fig:coverage}
\end{figure}

\subsection{Runtime evaluation}
In this case study, an age-structured population with $n = 6$ classes is considered. The population evolves according to the classical Leslie model~\cite{leslie1945use}
\begin{equation}
    x(k+1) = Ax(k)+Bu(k),
\end{equation}
with
\begin{equation}
A = 
\begin{bmatrix}
f_1 & f_2 & \cdots & f_{n-1} & f_n\\
s_1 & 0   & \cdots & 0 & 0   \\
0   & s_2 & \ddots & \vdots & \vdots \\
\vdots & \ddots & \ddots & 0 & 0\\
0 & \cdots & 0 & s_{n-1} & s_n
\end{bmatrix},\quad B = \begin{bmatrix} 0&0&0&0\\b_1&0&0&0\\0&b_2&0&0\\0&0&b_3&0\\0&0&0&b_4\\0&0&0&0
\end{bmatrix}.
\end{equation}
The state vector $x(k) = \left[x_1(k),\,x_2(k),\ldots,x_n(k)\right]^T$ represents the number of individuals in each age class at epoch $k$, while the input vector specifies the number of new individuals in the corresponding age classes. For simplicity, state and input values are intended as variations with respect to a
reference population, which explains 
why negative values may occur.
The coefficients $f_i$ denote the fertility rates of age class $i$, whereas $s_i$, $i = 1,\ldots n-1$ represent the survival rates from age class $i$ to age class $i+1$. The parameter $s_n$ denotes the fraction of individuals that survive in the last age class. In this example, the fertility and survival rates are nominally set to
$f_i = \{0.01,\,0.45,\,0.4,\,0.14,\,0,\,0\}$,  $s_i = \{0.9,\,0.95,\,0.9,\,0.85,\,0.8,\,0.2\}$. These parameters are affected by a relative uncertainty $\delta_f = 8\%$ on parameters $f_i$ and $\delta_s =3\%$ on parameters $s_i$. Also, the parameters $b_j$, $j = 1,\ldots,m$ are uncertain, with nominal value 1 and relative uncertainty $\delta_b = 1\%$.
The input vector must satisfy $ |u(k)| \leq \left[0.25,\,0.3,\,0.2,\,0.1\right]^T$, while all the state components are bounded as $|x_i(k)|\leq 1.5$. These constraints define the sets $\XX$ and $\UU$ in~\eqref{eq:constraint}.
% $$ \begin{bmatrix}
%     0\\0\\0\\0
% \end{bmatrix} \leq u(k) \leq \begin{bmatrix}
%        0.35\\ 0.35\\ 0.3\\ 0.25
% \end{bmatrix}.$$
% The number of new individuals that can be introduced into the population is limited.
% In particular, the input vector must satisfy
% $$ \begin{bmatrix}
%     0\\0\\0\\0
% \end{bmatrix} \leq u(k) \leq \begin{bmatrix}
%        0.35\\ 0.35\\ 0.3\\ 0.25
% \end{bmatrix}.$$
% Moreover, the total population must remain below a prescribed capacity limit, 
% corresponding to the state constraint
% \begin{equation}\label{eq:totalpop}
% \sum_{i=1}^6 x_i(k) \leq 6.1    
% \end{equation}
%The constraints $x_i(k) \geq 0$ are implicitly satisfied, given the model parameters and the input constraints. The control objective is to reach a steady state $x_{ref} = \left[1,\ldots,1\right]^T$ satisfying the input and state constraints defined above, for all possible value of the uncertain parameters $f_i$ and $s_i$. The control input related to the steady state is $u_{ref} = \left[0.1,\,0.05,\,0.1,\,0.15\right]^T$. 

For IM-MPC, the matrix gain $K$ is chosen so that the eigenvalues of $\AKc$ are $\{0.3,0.3,0.3,0.3,0.2,0.2\}$. The resulting  $K$ satisfies Assumption~\ref{assum:K}. The terminal set $\XX_f$ is calculated as in Section~\ref{subsec:toy} and $N_{\max} =10$ is set.
The proposed control law is compared to the PT-MPC and SLS-MPC algorithms with $N=10$, all equipped with the same terminal set $\XX_f$.

To evaluate the computational burden of the three methods, a sequence of simulation runs are carried out. In each run, a progressively larger subset of the 14 non-zero parameters ($f_i$, $s_i$ and $b_j$) are treated as uncertain. This is repeated for 20 initial conditions $-1\leq x_i(0) \leq 1$, for $i = 1,\ldots,6$.
%More specifically, the first run involves a single uncertain parameter, the second run two uncertain parameters, and so on, until all ten non-zero parameters are considered uncertain. 
Figure~\ref{fig:tocs} compares the number of constraints and the average solution time of problem~\eqref{eq:mpc_tube_box} with those of the other techniques. It can be observed that both the number of constraints and the solution time of IM-MPC are essentially unaffected by the number of uncertain parameters. 
{Conversely, these quantities increase exponentially for both PT-MPC and SLS-MPC. 
%This is due to the fact that the number of inequality constraints in the corresponding optimal control problems scales with the number of vertices of the uncertainty set~$\MD$. 
Notice that, despite having to solve a variable-horizon problem, the proposed approach is much less computationally demanding than SLS-MPC, for any number of uncertain parameters.} 
This example confirms that the proposed control law scales efficiently to systems of nontrivial dimension with many uncertain parameters.

% To further assess the robustness of the proposed control law, we carried out a stability–margin test in which the level of parametric uncertainty is progressively increased. In particular, starting from the nominal model (i.e., $\delta_f = \delta_s = 0$), the feasibility of problems~\eqref{eq:mpc_tube_box} and the one in~\cite{CANNONbook}, for $k=0$, are tested for increasing values of $\delta_f$ and $\delta_s$ up to 20\% of the nominal value. Figure~\ref{fig:ROA_population} shows that the proposed controller is initially feasible (and then robustly stable, see Theorem~\ref{thm:convergence})  for larger uncertainty levels than the method in~\cite{CANNONbook}, confirming its superior robustness properties, already observed in the simple example in Section~\ref{subsec:toy}.

% As a last result, we show the closed-loop trajectories obtained with a level of uncertainty $\delta_f = 10\%$ and $\delta_s = 4\%$, for 50 randomly generated pairs of matrices $\interval{A\quad B}\in\MD$, for both the compared methods. It can be seen that state and input constraints are robustly satisfied and that both controllers ensure asymptotic convergence. The trajectories in Figs.~\ref{fig:population_traj} are broadly comparable, although the proposed control law exhibits a faster convergence.

Figure \ref{fig:population_traj} shows the closed-loop trajectories of $x_1(k)$, $x_6(k)$ and $u(k)$, generated by the IM-MPC scheme starting from $x_i(0)=-1$ for all $i$, for 100 different realizations of the uncertain matrices $\interval{A\quad B}\in\MD$.
It can be observed that the trajectories converge to zero while robustly satisfying state and input constraints, according to the theoretical properties derived in Section \ref{sec:intervalMPC}.

%Combined with the reduced conservatism achieved over existing approaches, this enables the treatment of realistic, high-dimensional case studies, rather than the simple toy examples typically considered in the literature. \mycomm{questo paragrafo andrà rivisto alla luce di quello che mettiamo nella sezione..., l'ultima frase comunque la toglierei}}

\begin{figure}[h]
    \centering
    \includegraphics[width=0.49\linewidth]{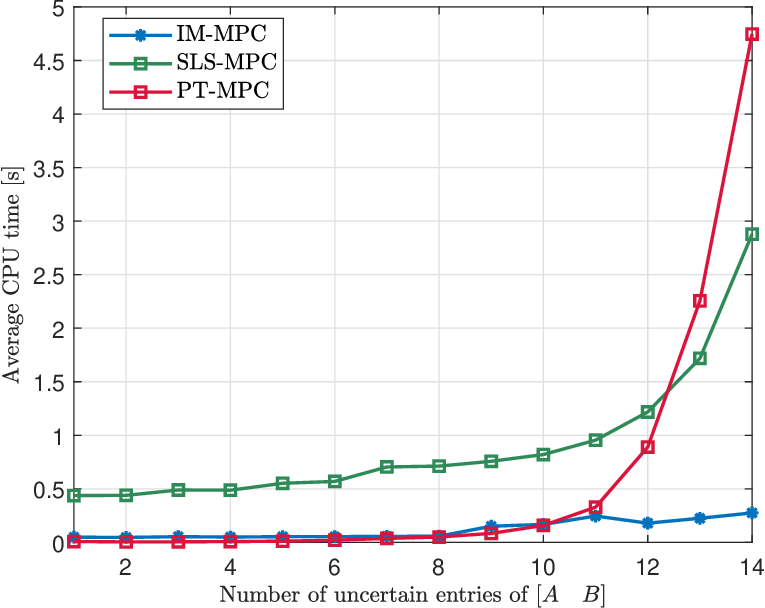}\hfill
    \includegraphics[width=0.49\linewidth]{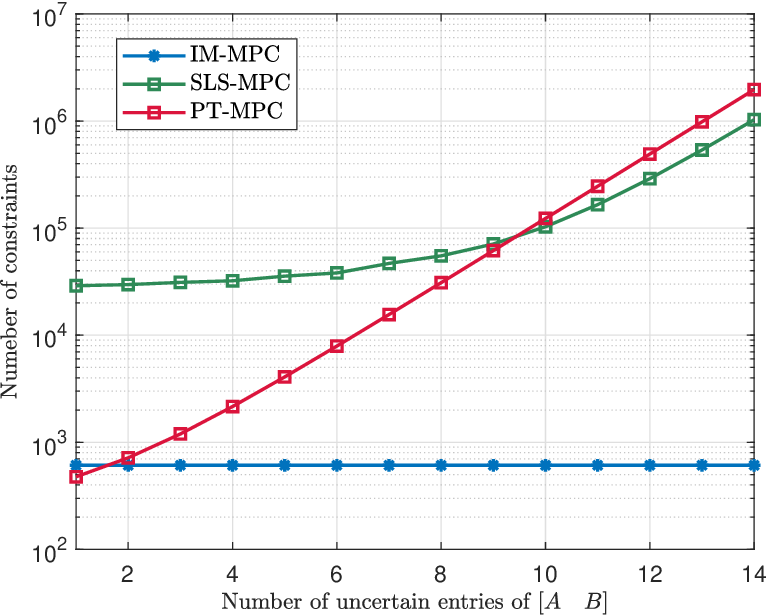}
    \caption{Average computation time required to solve the robust optimal control problem (left) and corresponding number of constraints (right, logarithmic scale) as functions of the number of uncertain entries in the matrix~$\left[A \;\; B\right]$.}
    \label{fig:tocs}
\end{figure}

\begin{figure}[h]
    \centering
    \includegraphics[width=0.7\linewidth]{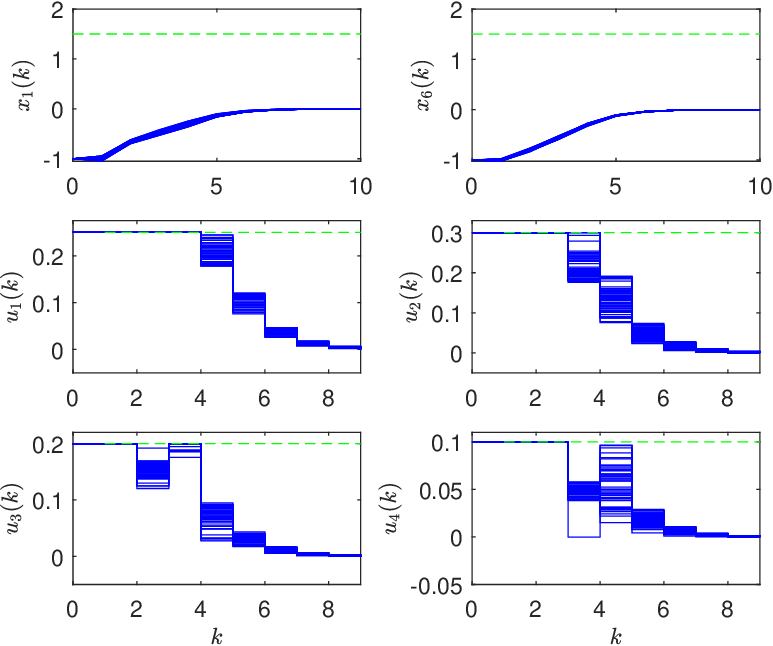}
    \caption{Closed-loop trajectories of states $x_1$ and $x_6$ (top) and control input (bottom), generated by IM-MPC, for different realizations of $\interval{A\;\;B} \in \MD$. Dashed lines represent constraints (green) and reference values (black).}
    \label{fig:population_traj}
\end{figure}

% \begin{figure}[h]
%     \centering
%     \includegraphics[width=0.7\linewidth]{figs/stability_margin_population.eps}
%     \caption{Feasibility of problem~\eqref{eq:mpc_tube_box} for different levels of parametric uncertainty $\delta_f,\delta_s$. \mycomm{cambiare legenda}}
%     \label{fig:ROA_population}
% \end{figure}

\section{Conclusions}
\label{sec:conclusions}

A new robust MPC technique has been presented, for discrete-time linear systems affected by parametric uncertainties. The key advantage of the proposed approach is that it exploits the uncertainty description in terms of interval matrices to devise a set-theoretic approximation of the system impulse response based on matrix zonotopes, which allows one to limit the conservatism of the resulting bound on the uncertainty propagation along the prediction horizon. This is confirmed by numerical simulations showing that the proposed approach matches the conservatism level of a recently developed technique based on system-level synthesis, which was shown to be significantly less conservative with respect to state-of-the-art robust MPC schemes.
Even more importantly, the proposed technique is extremely efficient in terms of online computational burden. In fact, it does not require vertex enumeration neither in the uncertainty model nor in the set approximating the uncertainty evaluation, thanks to the use of a suitable bounding operator exploiting the interval matrix uncertainty model. A variable-horizon optimal control problem is adopted to achieve recursive feasibility and to enable finite-time convergence to a robust control invariant set, thus ensuring robust asymptotic stability of the proposed scheme.
\\
Since uncertainty model based on interval matrices can be easily estimated from data within a set-membership identification setting, the extension of the proposed technique to data-driven robust MPC is the subject of ongoing work.
Future developments may involve the use of the zonotopic bounding procedure in combination with a time-varying nominal control law, as in the mentioned system-level synthesis approach, to further reduce conservatism.

\appendix
\section{Appendix}

\subsection{Auxiliary lemmas}

In order to prove the main results in the paper, the following preliminary lemmas are needed.
\begin{lemma}
\label{lem:Ij}
The set $\II(j)$ defined in \eqref{eq:Ibox} satisfies
\begin{equation}
\label{eq:Ijint}
\II(j)= \intsimple{ \sum_{i=0}^j | \hat{A}_K^{j-i} | F_i }
\end{equation}
where $F_0=\Delta_S$ and
\begin{equation}
\label{eq:Fj}
F_{j+1} = \Delta_K \sum_{i=0}^{j} | \hat{A}_K^{j-i} | F_i, ~~~\mbox{for}~j \geq 0. 
\end{equation}
\end{lemma}
\emph{Proof. } Let us first prove by induction that 
\begin{equation}
\label{eq:TTj}
\TT^j_{\MDK}(\II_{\Delta})=\langle 0; \hat{A}_K^j \EE(F_0),   \hat{A}_K^{j-1} \EE(F_1), \dots,  \EE(F_j) \rangle.
\end{equation}
For $j=1$, by applying Definition \ref{def:operatorT} with $\II=\MDK=\hat{A}_K \oplus \intsimple{\Delta_K}$ and $\MM=\II_{\Delta} = \langle 0;\,\mathcal{E}(\Delta_S)\rangle$ one gets
$$
\TT_{\MDK}(\II_{\Delta})=\langle 0; \hat{A}_K \EE(F_0),  \EE(F_1) \rangle,
$$
with $F_0=\Delta_S$ and $F_1=\Delta_K \Delta_S$, which is in agreement with \eqref{eq:Fj}. Now, let \eqref{eq:TTj} hold for a generic $j$. By applying again Definition \ref{def:operatorT}  one obtains
$$
\TT^{j+1}_{\MDK}(\II_{\Delta}) = \TT_{\MDK} \circ \TT_{\MDK}^j (\II_{\Delta}) = 
\langle 0; \hat{A}_K^{j+1} \EE(F_0),   \hat{A}_K^{j} \EE(F_1), \dots,  \hat{A}_K \EE(F_{j}), \EE(F_{j+1}) \rangle,
$$
where, using \eqref{eq:Fgen} and \eqref{eq:AdecM}, one has
$$
\begin{array}{rcl}
F_{j+1} &=& \displaystyle{ \Delta_K \sum_{i=1}^{n(n+m)} \left(  |\hat{A}_K^j F_0^{(i)}|+  |\hat{A}_K^{j-1} F_1^{(i)}|+ \dots + | F_j^{(i)} | \right) }\\ 
&=& \Delta_K \left( |\hat{A}_K^j| F_0 + |\hat{A}_K^{j-1} | F_1 + \dots +  F_j \right)
%\sum _{i=0}^j  |\hat{A}_K^{j-i}| F_i
\end{array}
$$
which corresponds to  \eqref{eq:Fj}.
Then, \eqref{eq:Ijint} immediately follows by applying Proposition \ref{prop:box} to the right hand side of \eqref{eq:TTj}. \qed
\begin{lemma}
\label{lem:Pj}
Let us define matrices $P_j$ such that $P_1=I$ and
\begin{equation}
\label{eq:Pj}
P_{j+1} = | \hat{A}_K^j | + \sum_{h=1}^{j} P_h \Delta_K | \hat{A}_K^{j-h} |, ~~~\mbox{for}~j \geq 1. 
\end{equation}
Then, matrices $F_j$ in \eqref{eq:Fj} satisfy
\begin{equation}
\label{eq:FjPj}
F_{j} = \Delta_K P_j \Delta_S, ~~~\mbox{for}~j \geq 1. 
\end{equation}
\end{lemma}
\emph{Proof. } By recursively expanding the sum in \eqref{eq:Fj}, one gets that $F_j=\Delta_K \Sigma_{j-1} \Delta_S$, where $\Sigma_j$ is the sum of the $2^j$ terms of the form $|M_1^{q_1}||M_2^{q_2}| \dots |M_h^{q_h}|$,
for $1 \leq h \leq j$,
with $M_i \in \{\hat{A}_K,\Delta_K\}$, $M_i \neq M_{i+1}$,
and $1 \leq q_i \leq j$, $\sum_{i=1}^h q_i =j$.
Similarly, by recursively expanding the sum in \eqref{eq:Pj}, one gets $P_j=\Sigma_{j-1}$, thus giving \eqref{eq:FjPj}.
%\mycomm{la proof precedente era sbagliata, questa dovrebbe essere corretta ma con la cosa un po' antiestetica che scriviamo $|\Delta_K^{q_i}|$ anche se il valore assoluto non servirebbe... ma serve per gli altri termini con $\hat{A}_K^{q_i}$!}\qed 

\subsection{Proof of Theorem~\ref{thm:RF}}\label{app:RF}
For proving recursive feasibility of the optimal control problem~\eqref{eq:mpc_tube_box}, we demonstrate the existence of a feasible candidate solution at time $k+1$. To do this, we need to distinguish the cases in which the horizon length of the optimal solution at time $k$ is $N_k^* > 1$ (Case 1) or $N_k^* = 1$ (Case 2).
%Let $\bar{k}$ be the first time instant in which $N_{\bar{k}}^* = 1$. Then we will show that $N_k^* = 1$ for all $k>\bar{k}$.

\textbf{Case 1 ($N_k^* >1$):}
Let us consider the following candidate solution for step $k+1$ 
\begin{equation}\label{eq:zv_candidate}
    \begin{array}{l}
        \hat{v}_{k+1}(j) =
            v^*_k(j+1) + K\AKc^j\delta(k), \quad j = 0,\ldots,N^*_k-2,\\
        \hat{z}_{k+1}(j) =
            z^*_k(j+1) + \AKc^j\delta(k), \quad j = 0,\ldots,N^*_k-1
    \end{array}
\end{equation}
where
\begin{equation}
\label{eq:deltak}
\delta(k) = x(k+1)-\hat{A} x(k) - \hat{B} u(k) = A_{\Delta} x(k) + B_{\Delta} u(k)
\end{equation}
and $z_k^*(j)$, $v_k^*(j)$ are the optimal nominal states and inputs for problem~\eqref{eq:mpc_tube_box}. Note that the length of this solution is $N_k^*-1$.
It is worth stressing that $\delta(k)$ is a known vector at time $k+1$ (when the state $x(k+1)$ is available), and therefore also the candidate solution in \eqref{eq:zv_candidate} is known. Moreover, one has that
\begin{equation}
\label{eq:deltakset}
\delta(k)
 \in \Mdelta\left[\begin{array}{c}
         x(k)\\
         u(k)
    \end{array}\right] 
    = \intsimple{\Delta_S}  \myvec{z_k^*(0)}{v_k^*(0)}.
\end{equation}
Let us show that \eqref{eq:zv_candidate} satisfies all the constraints of problem~\eqref{eq:mpc_tube_box}.

\emph{Initial constraint.} The initial state of the candidate solution~\eqref{eq:zv_candidate} satisfies
\begin{equation*}
    \begin{split}
        \hat{z}_{k+1}(0) & = z^*_k(1)+\delta(k)\\ 
        &= \Ac z^*_k(0)+\Bc v^*_k(0)+A_{\Delta}x(k) + B_{\Delta}u(k)\\
        & = (\Ac + A_{\Delta})x(k) + (\Bc + B_{\Delta})u(k) = x(k+1)
    \end{split}
\end{equation*}
which corresponds to constraint \eqref{subeq:mpc_init} at time $k+1$.

\emph{Dynamic constraint.} By using~\eqref{eq:zv_candidate}, we have that
\begin{equation*}
    \begin{split}
        &\hat{z}_{k+1}(j+1)-\AKc^{j+1}\delta(k) = z^*_k(j+2)\\ 
        & = \Ac z_k^*(j+1) + \Bc v_k^*(j+1)\\
        &= \Ac\left(\hat{z}_{k+1}(j)-\AKc^j\delta(k)\right)+\Bc\left(\hat{v}_{k+1}(j)-K\AKc^j\delta(k)\right) \\
        &= \Ac\hat{z}_{k+1}(j)+\Bc \hat{v}_{k+1}(j) - \AKc^{j+1}\delta(k).
    \end{split}
\end{equation*}
Hence,  $\hat{z}_{k+1}(j+1) = \Ac \hat{z}_{k+1}(j)+\Bc \hat{v}_{k+1}(j)$ implying satisfaction of constraint~\eqref{subeq:mpc_dynamic}.

\emph{State constraint.}
Let us define 
\begin{equation}
\label{eq:BBkstar}
\BB_k^*(j) = \sum_{i=0}^{j-1} \II(j-i-1) \xi_k^*(i)
\end{equation}
and
\begin{equation}
\hat\BB_{k+1}(j) = \sum_{i=0}^{j-1} \II(j-i-1) \hat\xi_{k+1}(i)
\end{equation}
where
$$
\xi_k^*(i) = \myvec{z_k^*(i)}{v_k^*(i)}\,,~~~
\hat\xi_{k+1}(i) = \myvec{\hat{z}_{k+1}(i)}{\hat{v}_{k+1}(i)}.
$$
In order to prove that the candidate solution \eqref{eq:zv_candidate} satisfies constraint~\eqref{subeq:mpc_state}, it is sufficient to show that
\begin{equation}
\label{eq:setincl}
\hat{z}_{k+1}(j) \oplus \hat\BB_{k+1}(j) \subseteq z_k^*(j+1) \oplus \BB_k^*(j+1) 
\end{equation}
for $j=0,1,\dots,N_k^*-1$, because the right hand side of \eqref{eq:setincl} is included in $\XX$. 
By using the second equation in  \eqref{eq:zv_candidate}, one has that \eqref{eq:setincl} is equivalent to
\begin{equation}
\label{eq:setincl2}
\hat{A}_{K}^j \delta(k) \oplus \hat\BB_{k+1}(j) \subseteq \BB_k^*(j+1). 
\end{equation}
Exploiting \eqref{eq:zv_candidate}, the left hand side of \eqref{eq:setincl2} can be rewritten as \begin{equation*}
    \begin{split}
	& \hat{A}_{K}^j \delta(k) \oplus \hat\BB_{k+1}(j) \\
	& = \hat{A}_{K}^j \delta(k) \oplus \sum_{i=0}^{j-1} \II(j-i-1) \hat\xi_{k+1}(i) \\
	& = \hat{A}_{K}^j \delta(k) \oplus \sum_{i=0}^{j-1} \II(j-i-1) \xi_{k}^*(i+1) \oplus \sum_{i=0}^{j-1} \II(j-i-1) M_K \hat{A}_K^i \delta(k) 
   \end{split}
\end{equation*}
where $M_K=\myvec{I}{K}$.
Therefore, by using \eqref{eq:BBkstar}, \eqref{eq:setincl2} boils down to
\begin{equation}	
\label{eq:setincl3}
\hat{A}_{K}^j \delta(k) \oplus \sum_{i=0}^{j-1} \II(j-i-1) M_K \hat{A}_K^i \delta(k)  \subseteq \II(j) \xi_k^*(0). 
\end{equation}
Now, using Lemmas \ref{lem:Ij} and \ref {lem:Pj} and recalling that $\Delta_K = \Delta_S |M_K|$, the left hand side of \eqref{eq:setincl3} satisfies 
\begin{align}
& \hat{A}_{K}^j \delta(k) \oplus \sum_{i=0}^{j-1} \II(j-i-1) M_K \hat{A}_K^i \delta(k)  \notag \\
& \overset{\eqref{eq:Ijint}}{=} \hat{A}_{K}^j \delta(k) \oplus \sum_{i=0}^{j-1} \intsimple{ \sum_{h=0}^{j-1-i} | \hat{A}_K^{j-1-i-h} | F_h } M_K \hat{A}_K^i \delta(k)   \notag\\
& \overset{\eqref{eq:FjPj}}{=} \hat{A}_{K}^j \delta(k) \oplus \sum_{i=0}^{j-1} \intsimple{ | \hat{A}_K^{j-1-i} | \Delta_S + \sum_{h=1}^{j-1-i} | \hat{A}_K^{j-1-i-h} | \Delta_K P_h \Delta_S } M_K \hat{A}_K^i \delta(k)  \notag \\
& \overset{\eqref{eq:intxM}}{\subseteq} \hat{A}_{K}^j \delta(k) \oplus \sum_{i=0}^{j-1} \intsimple{ \left\{ | \hat{A}_K^{j-1-i} | \Delta_S + \sum_{h=1}^{j-1-i} | \hat{A}_K^{j-1-i-h} | \Delta_K P_h \Delta_S \right\}\, |M_K| \, |\hat{A}_K^i | } \delta(k)  \notag\\
& =  \hat{A}_{K}^j \delta(k) \oplus \sum_{i=0}^{j-1} \intsimple{ | \hat{A}_K^{j-1-i} | \Delta_K \, |\hat{A}_K^i | } \delta(k)  \oplus \sum_{i=0}^{j-1} \intsimple{ \sum_{h=1}^{j-1-i} | \hat{A}_K^{j-1-i-h} | \Delta_K P_h \Delta_K  \, |\hat{A}_K^i | } \delta(k) \notag \\
&= \hat{A}_{K}^j \delta(k) \oplus \sum_{i=0}^{j-1} \intsimple{ | \hat{A}_K^{j-1-i} | \Delta_K \, |\hat{A}_K^i | }  \delta(k) \oplus  \sum_{h=1}^{j-1}  \sum_{i=0}^{j-1-h} \intsimple{ | \hat{A}_K^{j-1-i-h} | \Delta_K P_h \Delta_K  \, |\hat{A}_K^i | } \delta(k)  \label{eq:lastlhs}
\end{align}
where the last two equalities exploit property \eqref{eq:intsum} and a reorganization of the sum indexes.
Similarly, the right hand side of \eqref{eq:setincl3} satisfies
\begin{align}
& \II(j) \xi_k^*(0) =  \intsimple{ \sum_{i=0}^{j} | \hat{A}_K^{j-i} | F_i }  \xi_k^*(0) \notag \\
& \overset{\eqref{eq:FjPj}}{=}  \intsimple{| \hat{A}_K^{j} | \Delta_S + \sum_{i=1}^{j} | \hat{A}_K^{j-i} |  \Delta_K P_i \Delta_S}  \xi_k^*(0)  \notag \\
& \overset{\eqref{eq:Pj}}{=}  \intsimple{| \hat{A}_K^{j} | \Delta_S + \sum_{i=1}^{j} | \hat{A}_K^{j-i} |  \Delta_K  | \hat{A}_K^{i-1} | \Delta_S  + \sum_{i=1}^{j} | \hat{A}_K^{j-i} | \Delta_K  \sum_{h=1}^{i-1} P_h \Delta_K | \hat{A}_K^{i-1-h} |  \Delta_S }  \xi_k^*(0)  \notag\\
& =  \intsimple{| \hat{A}_K^{j} | \Delta_S + \sum_{i=0}^{j-1} | \hat{A}_K^{j-1-i} |  \Delta_K  | \hat{A}_K^{i} | \Delta_S  + \sum_{i=0}^{j-1} | \hat{A}_K^{j-1-i} | \Delta_K  \sum_{h=1}^{i} P_h \Delta_K | \hat{A}_K^{i-h} |  \Delta_S }  \xi_k^*(0)  \notag \\
& =  \intsimple{| \hat{A}_K^{j} | \Delta_S + \sum_{i=0}^{j-1} | \hat{A}_K^{j-1-i} |  \Delta_K  | \hat{A}_K^{i} | \Delta_S  + \sum_{h=1}^{j-1}  \sum_{i=h}^{j-1} | \hat{A}_K^{j-1-i} | \Delta_K  P_h \Delta_K | \hat{A}_K^{i-h} |  \Delta_S }  \xi_k^*(0)  \notag \\
& = \intsimple{| \hat{A}_K^{j} | \Delta_S + \sum_{i=0}^{j-1} | \hat{A}_K^{j-1-i} |  \Delta_K  | \hat{A}_K^{i} | \Delta_S  + \sum_{h=1}^{j-1}  \sum_{i=0}^{j-1-h} | \hat{A}_K^{j-1-i-h} | \Delta_K  P_h \Delta_K | \hat{A}_K^{i} |  \Delta_S }  \xi_k^*(0).   \label{eq:lastrhs}
\end{align}
Now, recall that from \eqref{eq:deltakset} one has $\delta(k) \in \intsimple{\Delta_S}\xi_k^*(0)$. Therefore, by comparing \eqref{eq:lastlhs} and \eqref{eq:lastrhs}, and using properties \eqref{eq:intprod} and \eqref{eq:intxMleft}, one gets that \eqref{eq:setincl3} holds, and therefore also \eqref{eq:setincl} and \eqref{subeq:mpc_state}.

\emph{Input constraint.} The proof that solution~\eqref{eq:zv_candidate} satisfies the input constraint~\eqref{subeq:mpc_input} follows the same arguments used for the state constraint \eqref{subeq:mpc_state}.

\emph{Terminal constraint.} Feasibility of the terminal constraint~\eqref{subeq:mpc_terminal} directly stems from \eqref{eq:setincl}. In particular,
$$
\hat{z}_{k+1}(N_k^*-1) \oplus \hat{\BB}_{k+1}(N_k^*-1) \subseteq z_k^*(N_k^*)\oplus \BB^*_k(N_k^*) \subseteq \XX_f.
$$

\textbf{Case 2 ($N_k^* = 1$):}
For time instant $k+1$, let us consider the candidate solution
\begin{equation}\label{eq:candidate2}
\begin{array}{l}
     \hat{v}_{k+1}(0) = Kx(k+1),  \\
     \hat{\textbf{z}}_{k+1} = \left(x(k+1),\,\AKc x(k+1)\right). 
\end{array}
\end{equation}
with associated horizon length $\hat{N}_{k+1} = N_k^* =1$.
Initial and dynamic constraints~\eqref{subeq:mpc_init}-\eqref{subeq:mpc_dynamic} are trivially satisfied by the candidate solution~\eqref{eq:candidate2}.
Now, let us show that $x(k+1)\in \XX_f$. From optimality of problem \eqref{eq:mpc_tube_box} (whose optimal horizon length is $N_k^*=1$), we know that $z_k^*(1) \oplus \BB_k^*(1) \subseteq \XX_f$. 
Therefore
\begin{equation}
\label{eq:xkinXf}
    \begin{split}
        x(k+1) &= A x(k) + B u(k) = \Ac x(k) + \Bc u(k) + \delta(k) \\
               & = z^*_{k}(1) + \delta(k) \in z^*_{k}(1) \oplus \Mdelta\myvec{z^*_k(0)}{v^*_k(0)}\\
               & = z_k^*(1) \oplus \BB_k^*(1) \subseteq \XX_f.
    \end{split}
\end{equation}
Hence, satisfaction of state constraint \eqref{subeq:mpc_state} stems from condition~\eqref{subeq:stateconstr}. 
Now, recalling that $\BB_{k+1}(0) = \{0\}$ and condition~\eqref{subeq:inputconstr}, one has
$$
\hat{v}_{k+1}(0) = Kx(k+1) \in K\XX_f \subseteq \UU.
$$
Therefore, the candidate solution~\eqref{eq:candidate2} is compliant with the input constraint~\eqref{subeq:mpc_input}. Concerning the terminal constraint, by using condition~\eqref{eq:RPI}, we have that
\begin{equation*}
\begin{split}
    \hat{z}_{k+1}(1)\oplus \hat{\BB}_{k+1}(1) & = \AKc x(k+1) \oplus \Mdelta\myvec{x(k+1)}{Kx(k+1)} \\
    & = \AKc x(k+1) \oplus \Mdelta \myvec{I}{K}x(k+1) \\
    & = \left( \AKc \oplus \Mdelta\myvec{I}{K}\right)x(k+1)\\
    & = \AK x(k+1) \subseteq \AK\XX_f \subseteq \XX_f.
\end{split}
\end{equation*}
Hence, the terminal constraint~\eqref{subeq:mpc_terminal} at step $k+1$ is verified. This concludes the proof. \qed

\subsection{Proof of Theorem~\ref{thm:convergence}}\label{app:convergence}

As in the proof of Theorem~\ref{thm:RF}, we treat separately the cases $N_k^* > 1$ and $N_k^* = 1$.

\textbf{Case 1 ($N_k^* > 1$):} The cost $\hat{J}_{k+1}$ associated to the candidate solution~\eqref{eq:zv_candidate} satisfies 
\begin{equation*}
    \begin{split}
        \hat{J}_{k+1} & = \gamma\hat{N}_{k+1}+\sum_{j=0}^{\hat{N}_{k+1}-1} \ell\left(\hat{v}_{k+1}(j)-K\hat{z}_{k+1}(j)\right)\\
        & = \gamma (N_k^*-1) + \sum_{j=0}^{N^*_k-2} \ell\left({v}^*_{k}(j+1)-K{z}^*_{k}(j+1)\right)\\
        & =\gamma N_k^*-\gamma + \sum_{j=1}^{N^*_k-1} \ell\left({v}^*_{k}(j)-K{z}^*_{k}(j)\right)\\
        & = J^*_k -\gamma - \ell\left({v}^*_{k}(0)-K{z}^*_{k}(0)\right)\leq J^*_k - \gamma.
    \end{split}
\end{equation*}
Hence, the optimal cost at time $k+1$ satisfies
\begin{equation}\label{eq:cost_decrease}
    J_{k+1}^* \leq \hat{J}_{k+1}\leq J_k^* - \gamma.
\end{equation}
This inequality ensures the existence of a time instant $\bar{k} \leq \lfloor J_0^*/\gamma\rfloor$ for which the optimal horizon length satisfies $N_{\bar{k}}^*=1$, because the minimum attainable cost is $J_k = \gamma$. Then, noting that $x(\bar{k}+1) \in \XX_f$, as shown in~\eqref{eq:xkinXf}, item (i) of Theorem~\ref{thm:convergence} follows.

\textbf{Case 2 ($N_k^* =1$):} The cost associated with the candidate solution~\eqref{eq:candidate2} is trivially $\hat{J}_{k+1} = \gamma$, which implies that such a solution is also optimal. Consequently, $J^*_k = \gamma$ for all $k > \bar{k}$. This condition indicates that, once the optimal horizon length equals one (i.e., at step $\bar{k}$), the minimizer of problems~\eqref{eq:mpc_tube_box} for all $k > \bar{k}$ coincides with the candidate solution~\eqref{eq:candidate2}.
Then, as a direct consequence, since the uncertain system~\eqref{eq:sys}-\eqref{eq:sys2} governed by the control law $u(k) = Kx(k)$ is robustly asymptotically stable for all $\left[A\quad B\right] \in \MD$ (see Assumption~\ref{assum:K}), the closed-loop system~\eqref{eq:sys}-\eqref{eq:sys2} with the MPC law \eqref{eq:u_RMPC},\eqref{eq:mpc_tube_box}-\eqref{eq:optvsel} is also robustly asymptotically stable. \qed

\bibliographystyle{ieeetr}
\bibliography{biblioVH}

\end{document}